\documentclass[conference,10pt]{IEEEtran}

\usepackage{psfrag}
\usepackage{amsmath}
\usepackage{amssymb}
\usepackage{amsfonts}
\usepackage{graphicx}
\usepackage{float}
\usepackage{graphics}
\usepackage{graphicx}
\usepackage{theorem}
\usepackage{euscript}

\newcommand{\mysize}{\small}
\newcommand{\myspace}{0cm}

\newcommand{\vct}[1]{\boldsymbol{#1}}

\theoremstyle{plain}
\newtheorem{thm}{Theorem}

\newtheorem{prop}[thm]{Proposition}

\begin{document}

\title{Cellular Systems with Full-Duplex Compress-and-Forward Relaying and Cooperative Base Stations}

\author{\authorblockN{Oren Somekh\authorrefmark{1}, Osvaldo Simeone\authorrefmark{2},
H. Vincent Poor\authorrefmark{1}, and Shlomo Shamai
(Shitz)\authorrefmark{3}}
\authorblockA{\authorrefmark{1}
Department of Electrical Engineering, Princeton University,
Princeton, NJ 08544, USA, {\{orens, poor\}@princeton.edu}}
\authorblockA{\authorrefmark{2}
CWCSPR, Department of Electrical and Computer Engineering, NJIT,
Newark, NJ 07102, USA, {osvaldo.simeone@njit.edu}}
\authorblockA{\authorrefmark{3}
Department of Electrical Engineering, Technion, Haifa 32000,
Israel, {sshlomo@ee.technion.ac.il}}}


\maketitle

\begin{abstract}
In this paper the advantages provided by multicell processing of
signals transmitted by mobile terminals (MTs) which are received
via dedicated relay terminals (RTs) are studied. It is assumed
that each RT is capable of full-duplex operation and receives the
transmission of adjacent relay terminals. Focusing on intra-cell
TDMA and non-fading channels, a simplified relay-aided uplink
cellular model based on a model introduced by Wyner is considered.
Assuming a nomadic application in which the RTs are oblivious to
the MTs' codebooks, a form of distributed compress-and-forward
(CF) scheme with decoder side information is employed. The
per-cell sum-rate of the CF scheme is derived and is given as a
solution of a simple fixed point equation. This achievable rate
reveals that the CF scheme is able to completely eliminate the
inter-relay interference, and it approaches a ``cut-set-like''
upper bound for strong RTs transmission power. The CF rate is also
shown to surpass the rate of an amplify-and-forward scheme via
numerical calculations for a wide range of the system parameters.
\end{abstract}

\vspace{\myspace}
\section{Introduction}
\vspace{\myspace}

Techniques for providing high data rate services and better
coverage in cellular mobile communications are currently being
investigated by many research groups. In this paper, we study the
combination of two cooperation-based technologies that are
promising candidates for achieving such goals, extending previous
work in \cite{Simeone-Somekh-BarNess-Spagnolini-WCOM07}\
-\nocite{Simeone-Somekh-Barness-Spagnolini-IT08}\nocite{Simeone-Somekh-Barness-Poor-Shamai-Allerton07}
\cite{Somekh-Simeone-Poor-Shamai-ISIT2007}. The first is relaying,
whereby the signal transmitted by a mobile terminal (MT) is
forwarded by a dedicated relay terminal (RT) to the intended base
station (BS) \cite{Lin-Hsu-Infocom00}.
The second technology of
interest here is multicell processing (MCP), which allows the BSs
to jointly decode the received signals, equivalently creating a
distributed receiving antenna array
\cite{Zhou-Zhao-Xu-Yao-COMMAG03}. The performance gain provided by
this technology within a simplified cellular model was first
studied in \cite{Wyner-94}
, under the assumption that BSs are
connected by an ideal backbone (see
\cite{Somekh-Simeone-Barness-Haimovich-Shamai-BookChapt-07} for a
survey on MCP).

Recently, the interplay between these two technologies has been
investigated for amplify-and-forward (AF) and decode-and-forward
(DF) protocols in
\cite{Simeone-Somekh-BarNess-Spagnolini-WCOM07}\cite{Somekh-Simeone-Poor-Shamai-ISIT2007}
and
\cite{Simeone-Somekh-Barness-Spagnolini-IT08}\cite{Simeone-Somekh-Barness-Poor-Shamai-Allerton07},
respectively. The basic framework employed in these works is the
Wyner uplink cellular model introduced in \cite{Wyner-94}.
According to the linear variant of this model, cells are arranged
in a linear geometry and only adjacent cells interfere with each
other. Moreover, inter-cell interference is described by a single
parameter $\alpha\in[0, 1]$, which defines the gain experienced by
signals travelling to interfered cells. Notwithstanding its
simplicity, this model captures the essential structure of a
cellular system and it provides insight into the system
performance.

\begin{figure}[tb]
\begin{center}
\psfrag{A\r}{\scriptsize$\alpha$}\psfrag{B\r}{\scriptsize$\beta$}
\psfrag{G\r}{\scriptsize$\gamma$} \psfrag{M\r}{\scriptsize$\mu$}
\psfrag{E\r}{\scriptsize$\eta$}
\includegraphics[scale=0.45]{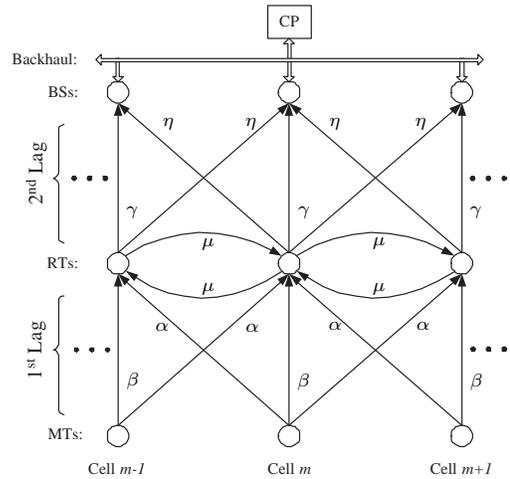}
\end{center}
\vspace{-0.5cm} \caption{Schematic diagram.} \label{fig: Wyner
relay network}
\end{figure}

In this work we adopt a similar setup to the one presented in
\cite{Somekh-Simeone-Poor-Shamai-ISIT2007}, in which dedicated
full-duplex (FD) RTs are added to the basic linear Wyner uplink
channel model and the signal path between adjacent RTs is
considered (i.e., inter-relay interference). 
With coverage extension in mind, we focus on distant users having no direct connection to the BSs.
Assuming a nomadic
application in which the RTs are oblivious to the MTs' codebooks,
a form of distributed compress-and-forward (CF) scheme with
decoder side information, similar to that of \cite{Gastpar-IT04}, is analyzed. It is noted that this scheme
resembles the single-user multiple-relay CF scheme considered in
\cite[Thm. 3]{Kramer-Gastpar-Gupta-IT05}. Focusing on a setup with
infinitely large number of cells, the achievable per-cell sum-rate
of the CF scheme is derived using the methods applied in \cite{Sanderovich-Somekh-Shamai-ISIT07}.
The achievable rate, which is
given as a solution of a simple fixed point equation, shows that
the CF scheme completely eliminates the inter-relay interferences.
Moreover the rate is shown to approach a ``cut-set-like'' upper
bound for strong RTs transmission power. Finally, the performance
of the CF scheme is compared numerically with the performance of
an AF scheme, recently reported in
\cite{Somekh-Simeone-Poor-Shamai-ISIT2007}, revealing the
superiority of the CF scheme for a wide range of the system
parameters.

\vspace{\myspace}
\section{System Model}
\vspace{\myspace}

We consider the uplink of a cellular system with a dedicated RT
for each transmitting MT. We focus on a scenario with no fading
and adopt a circular version of the linear cellular uplink channel
presented by Wyner \cite{Wyner-94}. RTs are added to the basic
Wyner model following the analysis in
\cite{Somekh-Simeone-Poor-Shamai-ISIT2007} (see Fig. \ref{fig:
Wyner relay network} for a schematic diagram of a single cell
within the setup and its inter-cell interaction).

The system includes $M$ identical cells arranged on a circle, with
a single MT active in each cell at a given time (intra-cell TDMA
protocol), and a dedicated single RT to relay the signals from the
MT to the BS (there is no direct connection between MTs and BSs).
Accordingly, each RT receives the signals of the local MT, the two
adjacent MTs, and the two adjacent RTs, with channel power gains
$\beta^2$, $\alpha^2$, and $\mu^2$ respectively. Likewise, each BS
receives the signals of the local RT, and the two adjacent RTs,
with channel power gains $\eta^2$ and $\gamma^2$ respectively. The
received signals at the RTs and BSs are affected by i.i.d.
zero-mean complex Gaussian additive noise processes with powers
$\sigma_1^2$ and $\sigma_2^2$, respectively. It is assumed that
the MTs use independent randomly generated complex Gaussian
codebooks with zero-mean and power $P$, whereas the RTs are
subjected to an average transmit power constraint $Q$. The RTs are
assumed to be oblivious of the MTs codebooks (nomadic
application), and that no cooperation among MTs is allowed. In
addition, the RTs are assumed to be capable of receiving and
transmitting simultaneously (i.e., perfect echo-cancellation). It
is noted that the propagation delays between the different nodes
of the system are negligible with respect to the symbol duration.
Finally, it is assumed that the BSs are connected to a central
processor (CP) via an ideal backhaul network, and that the channel
path gains and noise powers are known to the BSs, MTs, and CP.


\vspace{\myspace}
\section{Preliminaries}
\vspace{\myspace}
\subsection{Wyner's Model - Sum-Rate Capacities}\label{sec: Wyner
capacities} \vspace{\myspace}

Putting aside the inter-relay interference paths
and the lack of joint MCP among the RTs, the mesh network of Fig.
\ref{fig: Wyner relay network} is composed of two Wyner models (or
two ``Wyner lags") \cite{Wyner-94}. The close relations of the
current setup and the Wyner model renders the following
definitions useful in the sequel.

The per-cell sum-rate capacity of the linear (or circular) Wyner
uplink channel with infinitely large number of cells (${\mysize
M\rightarrow\infty}$), no user cooperation, optimal MCP,
signal-to-noise ratio (SNR) $\rho$, inter-cell interference factor
$a$ (e.g. $\alpha$ or $\eta$ in Fig. \ref{fig: Wyner relay
network}), and local path gain $b$ (e.g. $\beta$ or $\gamma$ in
Fig. \ref{fig: Wyner relay network}), is given by \cite{Wyner-94}
\begin{equation}\mysize\label{eq: Wyner rate no WF}
R_{\mathrm{w}}(a,b,\rho)=\int_0^1\log_2\left(1+\rho
H(f)^2\right)df\ .
\end{equation}
where $H(f)=b+2a\cos 2\pi f$. When transmitter full cooperation is
allowed the per-cell sum rate capacity of the above channel is
achieved by ``waterfilling'' solution and is given by
\begin{equation}\mysize\label{eq: Wyner rate WF}
\begin{aligned}
R^{\mathrm{wf}}_{\mathrm{w}}(a,b,\rho)&=
\int_0^1\log_2\left(1+\left(\nu-\frac{1}{H(f)^2}\right)^+H(f)^2\right)df\\
&\quad\mathrm{s.t.}\
\int_0^1\left(\nu-\frac{1}{H(f)^2}\right)^+=\rho\ ,
\end{aligned}
\end{equation}
where $(x)^+=\min\{x,0\}$.

\vspace{\myspace}
\subsection{Upper Bound}
\vspace{\myspace}

Denoting $\rho_1=P/\sigma^2_1$ and $\rho_2=Q/\sigma^2_2$ as the
SNRs over the first ``MT-RT" and second ``RT-BS" lags,
respectively, we have the following bound. \vspace{-0.3cm}
\begin{prop}
The per-cell sum-rate of any scheme employed in the relay-aided
Wyner circular uplink channel with infinite number of cells
${\mysize M\rightarrow\infty}$ and no MT cooperation, is upper
bounded by
\begin{equation}\mysize\label{eq: Cut-set bound}
    R_{\mathrm{ub}} =
    \min\left\{R_\mathrm{w}(\alpha,\beta,\rho_1),R^{\mathrm{wf}}_\mathrm{w}(\eta,\gamma,\rho_2)\right\}\
    .
\end{equation}
\end{prop}
\vspace{\myspace}
\begin{proof}(outline)
The rate expression is easily derived by considering two cut-sets,
one separating the MTs from the RTs and the other separating the
RTs from the BSs (or CP). We refer to this bound as
``cut-set-like" bound since we also account for the assumption of
no MTs cooperation in the first lag.
\end{proof}
It is noted that the upper bound continues to hold even if we
allow multiple MTs to be simultaneously active in each cell
(assuming a total-cell transmit power of $P$). Since both
arguments of \eqref{eq: Cut-set bound} increase with SNR it is
easily verified that
$R_{\mathrm{ub}}\underset{\rho_1\rightarrow\infty}{\rightarrow}R^{\mathrm{wf}}_\mathrm{w}(\eta,\gamma,\rho_2)$
and that
$R_{\mathrm{ub}}\underset{\rho_2\rightarrow\infty}{\rightarrow}R_\mathrm{w}(\alpha,\beta,\rho_1)$.

\vspace{\myspace}
\subsection{Amplify-and-Forward Scheme}
\vspace{\myspace}

As a reference result, we consider the AF scheme with MCP analyzed
in \cite{Somekh-Simeone-Poor-Shamai-ISIT2007} for a similar
infinite setup. For the AF scheme we make an additional assumption
regarding the relaying delay, namely that the RTs amplify and
forward the received signals with a delay of $\lambda\ge 1$
symbols (an integer). Interpreting the cellular uplink channel
model with AF and MCP as a 2D linear time invariant system, and
applying the 2D extension of Szeg{\"o}'s theorem \cite{Wyner-94},
the following result is derived in
\cite{Somekh-Simeone-Poor-Shamai-ISIT2007}. \vspace{-0.3cm}
\begin{prop}\label{prop: MCP AF sum-rate}
An achievable per-cell sum-rate of AF relaying with optimal MCP
and no spectral shaping, employed in the relay-aided infinite
linear Wyner uplink channel, is given by
\begin{equation}\mysize\label{eq: MCP AF sum-rate}
R_{\mathrm{af}}=\int_0^{1}\log\left(\frac{A+B+\sqrt{(A+B)^2-C^2}}{B+\sqrt{B^2-C^2}}\right)df\
,
\end{equation}
where
\begin{equation*}\mysize\label{eq: A B C definitions}
\begin{aligned}
A&\triangleq P g^2(\beta+2\alpha\cos2\pi f)^2(\gamma+2\eta\cos2\pi f)^2\\
B&\triangleq\sigma_1^2 g^2(\gamma+2\eta\cos2\pi f)^2+\sigma_2^2(1+4g^2\mu^2\cos^22\pi f)\\
C&\triangleq 4\sigma_2^2 g\mu\cos2\pi f\ .
\end{aligned}
\end{equation*}
Furthermore, the optimal relay gain $g_\mathrm{o}$ is the unique
solution to the equation $\sigma^2_r(g)=Q$ where
\begin{equation}\mysize\label{eq: MCP relay power}
    \sigma^2_r(g)=\frac{(P\beta^2+\sigma^2_1)g^2}{\sqrt{1-(2\mu
g)^4}}+\frac{4P\alpha^2 g^2}{\sqrt{1-(2\mu g)^2}+1-(2\mu g)^2}
\end{equation}
is the relay output power.
\end{prop}\vspace{-0.2cm}
It is shown in \cite{Somekh-Simeone-Poor-Shamai-ISIT2007} that the
optimal gain is achieved when the relays use their full power $Q$,
and that
$g_\mathrm{o}\underset{Q\rightarrow\infty}{\longrightarrow}
1/(2\mu)$. In addition, $R_{\mathrm{af}}$ is not interference
limited and it is independent of the actual RT delay value
$\lambda$.

\vspace{\myspace}
\section{Distributed Compress-and-Forward Scheme}\label{sec: Distributed Compress-and-Forward
Scheme}\vspace{\myspace}

\begin{figure}[tb]
\begin{center}
\includegraphics[scale=0.41]{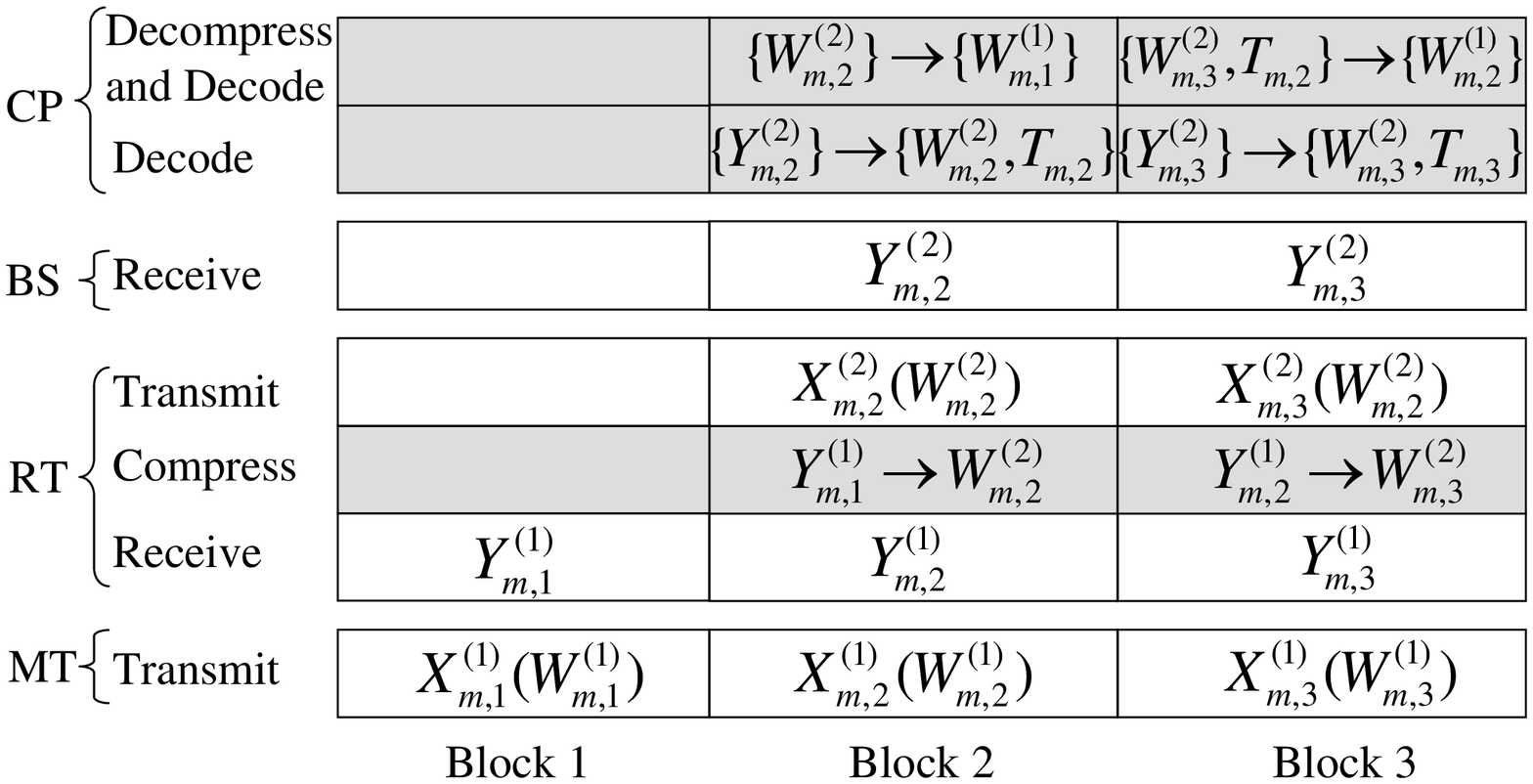}
\end{center}
\vspace{-0.3cm} \caption{CF scheme diagram.} \label{fig: CF
scheme}
\end{figure}

Here we describe the proposed CF-based transmission scheme, which
organizes transmission into successive blocks (or codewords) of
$N$ symbols, as sketched in Fig. \ref{fig: CF scheme}. It should
be remarked that transmission in the AF scheme presented in the
previous section spans only one block (with some $o(N)$ symbols
margin due to the delay $\lambda$ and the filter effective
response time). For this reason, while in the AF scheme the RTs
need to maintain only symbol synchronization, for the CF scheme to
be discussed below, block synchronization is also necessary.


With ${{\small (\cdot)^{(1)},\ (\cdot)^{(2)}}}$ denoting the
association to the first \textquotedblleft MT-RT" and second
\textquotedblleft RT-BS" lags, respectively, the received signal
at the $m$th RT in an arbitrary symbol of the $n$th block is
\begin{equation}
{\small Y_{m,n}^{(1)}=\beta X_{m,n}^{(1)}+\alpha(X_{[m-1],n}^{(1)}%
+X_{[m+1],n}^{(1)})+T_{m,n}+Z_{m,n}^{(1)}\ ,}\label{eq: RT received signal}%
\end{equation}
where ${{\small [k]\triangleq k\ \mathrm{mod}\ M}}$,
$X_{m,n}^{(1)}$ are the signals transmitted by the MTs (to be
defined in the sequel), $Z_{m,n}^{(1)}$ denotes the additive noise
at the RT, and the inter-relay interference is
\begin{equation}
{\small T_{m,n}=\mu(X_{[m-1],n}^{(2)}+X_{[m+1],n}^{(2)})\ .}%
\label{eq: Inter-relay interference}%
\end{equation}
The received signal at the $m$th BS is
\begin{equation}
{\small Y_{m,n}^{(2)}=\gamma X_{m,n}^{(2)}+\eta(X_{[m-1],n}^{(2)}%
+X_{[m+1],n}^{(2)})+Z_{m,n}^{(2)}\ ,}%
\end{equation}
where $X_{m,n}^{(2)}$ are the signals transmitted by the RTs, and
$Z_{m,n}^{(2)}$ denotes the additive noise at the BS.


The proposed CF scheme works as follows (see Fig. \ref{fig: CF
scheme} where shadowed boxes indicate zero time processing). The
basic idea is to have the RTs compress the signal $Y_{m,n}^{(1)}$
received in any $n$th block (say $n=2$ in Fig. \ref{fig: CF
scheme}) and forward it in the $(n+1)$th block (e.g., $n+1=3)$ via
a channel codeword $X_{m,n+1}^{(2)}$, by exploiting the side
information available at the CP about the compressed signals
$Y_{m,n}^{(1)}$. In fact, with the proposed scheme, in the $n$th
block, the CP decodes the channel codewords ${X_{m,n}^{(2)}}$
transmitted by the RTs, and these are correlated with the signal
$Y_{m,n}^{(1)}$ \eqref{eq: RT received signal} via $T_{m,n}$
\eqref{eq: Inter-relay interference}. Based on this side
information, distributed CF is implemented at the RTs according to
\cite{Gastpar-IT04} via standard vector quantization and binning.
A more formal description of the CF scheme is presented
below.\newline \textbf{Code Construction}: 1) \emph{At the MTs}:
each $m$th MT generates a rate-$R_{\mathrm{cf}}$ Gaussian random channel
codebook ${{\small \mathcal{X}_{m}^{(1)}}}$ according to ${{\small
\mathcal{CN}(0,\rho_{1})}}$ (no optimality is claimed);
2) \emph{At the RTs}: 2.a) Each RT generates a rate-${{\small R_{\mathrm{w}%
}(\eta,\gamma,\rho_{2})}}$ Gaussian random channel codebook
${{\small \mathcal{X}_{m}^{(2)}}}$ according to ${{\small \mathcal{CN}%
(0,\rho_{2})}}$;
2.b) Each RT generates a rate-${{\small
\hat{R}=I(Y^{(1)}_{m};U_{m})}}$ Gaussian quantization codebook
${{\small \mathcal{U}_{m}}}$ according to the marginal
distribution induced by
\begin{equation}
{\small U}_{m}{\small =Y}_{m}^{(1)}{\small +V}_{m}{\small ,}%
\label{eq: quantization}%
\end{equation}
where the quantization noises ${{\small V_{m}}}$ are i.i.d.
zero-mean complex Gaussian independent of all other random
variables (no optimality is claimed). Each quantization codebook
is randomly partitioned into ${{\small
2^{NR_{\mathrm{w}}(\eta,\gamma,\rho_{2})}}}$ bins, each of size
${{\small 2^{N(\hat{R}-R_{\mathrm{w}}(\eta,\gamma,\rho_{2}))}}}$%
;\newline\textbf{Encoding at the MTs}: each MT sends its message
${{\small W_{m,n}^{(1)}\in\mathcal{W}^{(1)}=\{1,\ldots,2^{NR_{\mathrm{cf}}}\}}}$
by
transmitting the $N$ symbol vector ${{\small \boldsymbol{X}_{m,n}^{(1)}}%
}=\mathcal{X}_{m}^{(1)}(W_{m,n}^{(1)})$ over the first
\textquotedblleft MT-RT" lag;\newline\textbf{Processing at the
RTs}: 1) \emph{Compressing}: each RT employs vector quantization
using standard joint typicality arguments via the quantization
codebok ${{\small \mathcal{U}_{m}}}$, to compress the
\emph{previously} received vector ${{\small
\boldsymbol{Y}_{m,n-1}^{(1)}}}$ into ${{\small
\boldsymbol{U}_{m,n}}}$ with the corresponding bin index ${{\small
W_{m,n}^{(2)}}}$; 2) \emph{Encoding}: each RT sends its bin index
${{\small W_{m,n}^{(2)}\in\mathcal{W}^{(2)}=\{1,\ldots,2^{NR_{\mathrm{w}}%
(\eta,\gamma,\rho_{2})}\}}}$ by transmitting ${{\small \boldsymbol{X}_{m,n}%
^{(2)}}}=\mathcal{X}_{m}^{(2)}(W_{m,n}^{(2)})$ over the second
\textquotedblleft RT-BS" lag;\newline\textbf{Decoding at the CP}:
1) \emph{Decoding the bin indices}: the CP collects the received
signal vectors
${{\small \boldsymbol{Y}_{\mathcal{M},n}^{(2)}}}$ (where ${{\small \mathcal{M}%
=\{0,1,\ldots,M-1\}}}$) from all the BSs through the backhaul
links. Then it decodes the resulting multiple-access channel (MAC)
using standard methods (e.g., \cite{Cover-Thomas-1991}) to recover
an estimate ${{\small \hat {W}_{\mathcal{M},n}^{(2)}}}$; 2)
\emph{Composing the side information}: the CP uses the decoded bin
indices ${{\small \hat{W}_{\mathcal{M},n}^{(2)}}}$ to compose the
side information vectors ${{\small
\boldsymbol{\hat{T}}_{\mathcal{M},n}}}$, where
${{\small \boldsymbol{\hat{T}}_{m,n}=\mu(\boldsymbol{\hat{X}}_{[m-1],n}%
^{(2)}+\boldsymbol{\hat{X}}_{[m+1],n}^{(2)})}}$, to be used in the
\emph{next} block; 3) \emph{Decoding the MTs messages}: The CP
uses the \emph{previous} side information
$\boldsymbol{\hat{T}}_{\mathcal{M},n-1}$ and looks for a unique
joint typical triplet ${{\small \{\boldsymbol{X}_{\mathcal{M},n-1}%
^{(1)},\boldsymbol{U}_{\mathcal{M},n},\boldsymbol{\hat{T}}_{\mathcal{M}%
,n-1}\}}}$ within the bins indicated by ${{\small \hat{W}_{\mathcal{M}%
,n}^{(2)}}}$, according to the joint distribution induced by
\eqref{eq: RT received signal}, to recover ${{\small \hat{W}_{\mathcal{M}%
,n-1}^{(1)}}}$.

\vspace{\myspace}
\section{Sum-Rate Analysis}
\vspace{\myspace}

Here we derive the per-cell sum-rate (or symmetric rate)
achievable via the proposed CF scheme.

\vspace{-0.3cm}
\begin{prop}\label{prop: CF rate}
An achievable per-cell sum-rate of the CF scheme employed in the
relay-aided Wyner circular uplink channel with infinite number of
cells ${\mysize M\rightarrow\infty}$, is given by
\begin{equation}\mysize\label{eq: CAF rate}
    R_{\mathrm{cf}}=R_{\mathrm{w}}\left(\alpha,\beta,\rho_1(1-2^{-r^*})\right)\ ,
\end{equation}
where $r^*\ge 0$ is the unique solution to the following fixed
point equation
\begin{equation}\mysize\label{eq: CAF FP equation}
R_{\mathrm{w}}\left(\alpha,\beta,\rho_1(1-2^{-r^*})\right)=R_{\mathrm{w}}(\eta,\gamma,\rho_2)-r^*\
.
\end{equation}
\end{prop}
\begin{proof}(outline) See Appendix \ref{appx: MCP AF sum-rate}.
\end{proof}
It is concluded that the rate ${\mysize R_{\mathrm{cf}}}$ is
\emph{independent} of the inter-relay interference. Moreover, the
CF scheme performs as if there are no inter-relay interferences
(i.e. ${\mysize\mu=0}$) and its rate coincides with the results of
\cite{Sanderovich-Somekh-Shamai-ISIT07} interpreting the second
``RT-BS" lag as the backhaul network with limited capacity
${\mysize C=R_{\mathrm{w}}(\eta,\gamma,\rho_2)}$. Also note, that
the result holds even if we relax the RT perfect echo-cancellation
assumption as long as the CP is aware of the residual echo power
gain.

Since $R_{\mathrm{w}}$ is given in an implicit integral form
\eqref{eq: Wyner rate no WF}, we can not solve the fixed point
equation \eqref{eq: CAF FP equation} analytically. Nevertheless,
since ${\mysize
R_{\mathrm{w}}\left(\alpha,\beta,\rho_1(1-2^{-r})\right)}$ is
monotonic in ${\mysize r}$, \eqref{eq: CAF FP equation} is easily
solved numerically. It is also evident that the CF rate increases
with the relay power ${\mysize Q}$. Hence, as with the AF scheme
full relay power usage is optimal.

It is easily verified that when ${\mysize\rho_1\rightarrow\infty}$
then ${\mysize r^*\rightarrow 0}$, and ${\mysize R_{\mathrm{cf}}}$
does not achieve the upper bound \eqref{eq: Cut-set bound}. This
is since ${\mysize
R_{\mathrm{cf}}\underset{\rho_1\rightarrow\infty}\rightarrow
R_{\mathrm{w}}(\eta,\gamma,\rho_2)\le
R^{\mathrm{wf}}_{\mathrm{w}}(\eta,\gamma,\rho_2)}$. On the other
extreme when ${\mysize\rho_2\rightarrow\infty}$ then ${\mysize
r^*\rightarrow\infty}$, and ${\mysize R_{\mathrm{cf}}}$ achieves
the upper bound ${\mysize
R_{\mathrm{cf}}\underset{\rho_2\rightarrow\infty}\rightarrow
R_{\mathrm{w}}(\alpha,\beta,\rho_1)}$.
In the next section, numerical results reveal that the CF scheme
outperforms the AF scheme over a wide range of the system
parameters.

\vspace{\myspace}
\section{Numerical Results}
\vspace{\myspace}

\begin{figure}[tb]
\begin{center}
\includegraphics[scale=0.42]{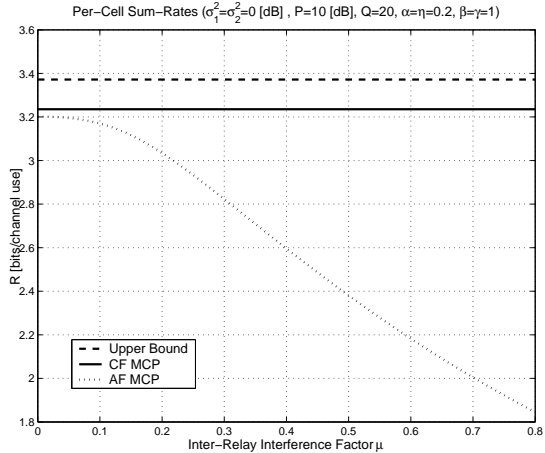}
\end{center}\vspace{-0.3cm}
\caption{Rates vs. the inter-relay interference factor $\mu$
(symmetrical hops).} \label{fig: Sum-rates vs mu}
\end{figure}

\begin{figure}[tb]
\begin{center}
\includegraphics[scale=0.42]{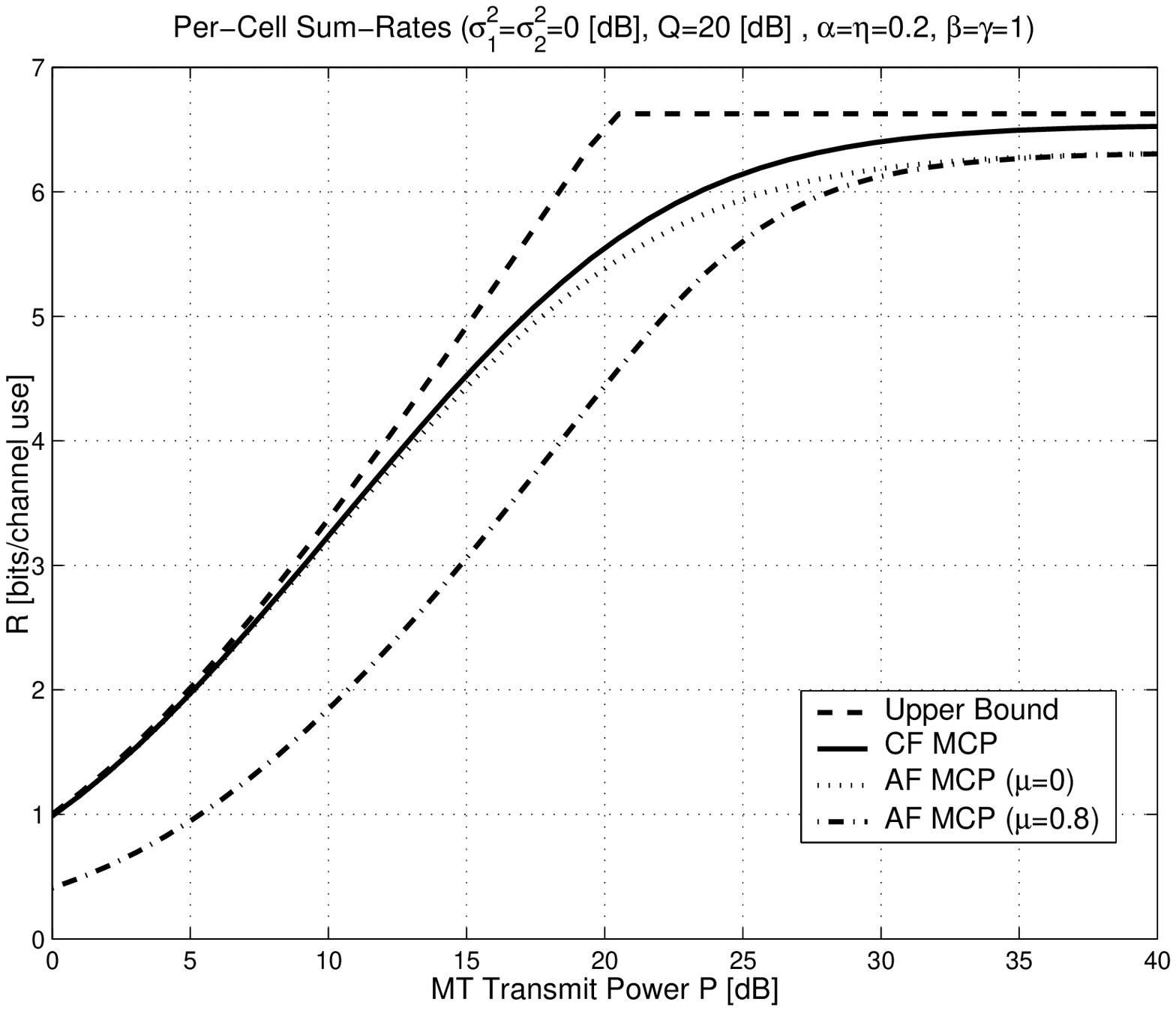}
\end{center}\vspace{-0.3cm}
\caption{Rates vs. the MTs transmission power $P$ (symmetrical
hops).} \label{fig: Sum-rates vs P sym}
\end{figure}

\begin{figure}[tb]
\begin{center}
\includegraphics[scale=0.42]{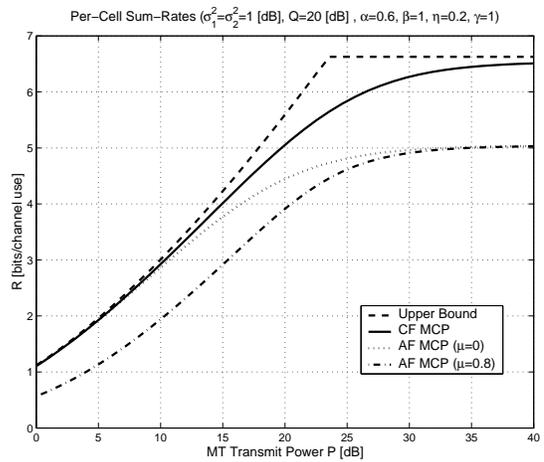}
\end{center}\vspace{-0.3cm}
\caption{Rates vs. the MTs transmission power $P$ (asymmetrical
hops).} \label{fig: Sum-rates vs P asym}
\end{figure}


In Fig. \ref{fig: Sum-rates vs mu} the per-cell sum-rates of the
CF and AF schemes are plotted along with the upper bound
\eqref{eq: Cut-set bound}, as functions of the inter-relay
interference factor ${\mysize\mu}$ for
${\mysize\rho_1=P/\sigma^2=10}$ [dB], ${\mysize\rho_2=Q/\sigma^2=
20}$ [dB], ${\mysize\sigma^2_1=\sigma^2_2=\sigma^2=1}$,
${\mysize\alpha=\eta=0.2}$, and ${\mysize\beta=\gamma=1}$. It is
noted that the AF curve is plotted with optimal relay gain
(resulting in a full usage of the relay power ${\mysize Q}$).
Examining the figure, the benefits of the CF scheme are evident in
view of the deleterious effect of increasing inter-relay interference
$\mu$ on the AF rate; the CF scheme provides almost twice the bits per
channel use than the AF rate, for strong inter-relay interference
levels. Also visible from the figure is the proximity of the CF
rate to the upper bound (less than 0.2 bits per channel use) which
for this setting is dominated by the rate of the first ``MT-RT"
lag.

Figures \ref{fig: Sum-rates vs P sym} and \ref{fig: Sum-rates vs P
asym} present the CF and AF rate curves and the upper-bound as
functions of the MTs power ${\mysize P}$, for
${\mysize\rho_2=Q/\sigma^2=20}$ [dB],
${\mysize\sigma^2_1=\sigma^2_2=1}$, and ${\mysize\beta=\gamma=1}$.
Here, we focus on two scenarios: (a) a setting with symmetrical
first and second lags, i.e. ${\mysize\alpha=\eta=0.2}$, and (b) a
setting with asymmetrical lags, i.e. ${\mysize\alpha=0.6,\
\eta=0.2}$. In both figures we include the AF rate curves for the two
extremes ${\mysize\mu=0}$ and ${\mysize\mu=0.8}$, which represent
weak and strong inter-relay interference scenarios, respectively.
It is noted that any AF rate curve with ${\mysize 0<\mu<0.8}$ is
confined between these two curves. Examining the figures it is
observed that the CF performs well (within one bit per channel use
of the upper bound) in both scenarios over the entire displayed
range of MTs power ${\mysize P}$. On the other hand, the AF scheme
performs well in both scenarios only for low inter-relay
interference levels and low MTs power ${\mysize P}$.
Other results (not presented here) show that the AF scheme is
slightly beneficial over the CF scheme under certain conditions
(e.g., high ${\mysize P}$, and asymmetrical lags
${\mysize\alpha=0.2,\ \eta=0.6}$).


\vspace{\myspace}
\section{Concluding Remarks}
\vspace{\myspace}

In this work we have considered a simplified two-hop cellular
setup with FD relays and inter-relay interference. Focusing on
nomadic application, a form of distributed CF with decoder side
information scheme, has been analyzed. We have derived the
achievable per-cell sum-rate for an infinitely large number of
cells, and have shown that the CF scheme totally eliminates the
inter-relay interference. Numerical results reveal that the CF
rate curves are rather close to a ``cut-set-like" upper bound, and
also demonstrate the superiority of the proposed CF scheme over
the MCP AF scheme of \cite{Somekh-Simeone-Poor-Shamai-ISIT2007},
for a wide range of the system parameters. \vspace{-0.1cm}
\appendix
\vspace{\myspace}
\subsection{Proof of Proposition \ref{prop: MCP AF sum-rate} (Outline)}\label{appx: MCP AF sum-rate}
\vspace{\myspace}

We focus on the decoding stage at the CP for the ${\small n}$th
block (recall Fig. \ref{fig: CF scheme}). Since the rate of the
channel codebooks used by the RTs on the second lag is equal to
the per-cell capacity ${\small R}_{w}$ of the corresponding Wyner
channel (see Sec. \ref{sec: Wyner capacities}), the CP is able to
correctly decode ${\small W}_{\mathcal{M},n-1}^{(2)}$ from the
previous block and ${\small W}_{\mathcal{M},n}^{(2)}$ from the
current with high probability. Based on the former, it can also
build an accurate estimate ${\small
\hat{\vct{T}}}_{\mathcal{M},n-1}.$ As per Fig. \ref{fig: CF
scheme}, the CP then attempts to decode the messages ${\small
W}_{\mathcal{M},n-1}^{(1)}.$ In the following, the variables of
interest are ${{\small Y_{m,n-1}^{(1)}}}$, ${{\small U_{m,n}}}$
and ${{\small X_{m,n-1}^{(1)}}}$ which are denoted for simplicity
as ${{\small Y_{m}}}$, ${{\small U_{m}}}$ and ${{\small X_{m}.}}$
To elaborate, we note that, due to the quantization rule (\ref{eq:
quantization}), the following Markov relation holds ${{\small
\{X_{\mathcal{M}},U_{\mathcal{M}\setminus
m},T_{\mathcal{M}}\}-Y_{m}-U_{m}.}}$ Recall also that the CP
decodes ${{\small X_{\mathcal{M}}}}$ by looking for jointly
typical sequences ${{\small
\{\mathbf{X}_{\mathcal{M}},\mathbf{U}_{\mathcal{M}},\mathbf{\hat{T}}_{\mathcal{M}}\},}}$
where ${{\small X_{\mathcal{M}}}}$ belong to the MTs codebooks
(each of size ${{\small 2^{NR_{\mathrm{cf}}}}}$) and ${{\small
U_{\mathcal{M}}}}$ are within the bins (of size ${{\small
2^{N(\hat{R}-R_{\mathrm{w}}(\eta ,\gamma,\rho_{2}))}}}$) whose
indices are given by ${\small W}_{\mathcal{M},n}^{(2)}$.

Assuming ${{\small \hat{R}\geq I(Y_{m};U_{m})}}$, for large block
length ${{\small N}}$, the probability of error is dominated by
the events where a set with erroneous $\mathbf{X}{{\small
_{\mathcal{L}}}}$ and $\mathbf{U}{{\small _{\mathcal{S}}}}$, for
any subsets ${{\small \mathcal{L},\mathcal{S\subseteq M}}}$, is
found that is jointly typical in the sense explained above (see
\cite{Sanderovich-Somekh-Shamai-ISIT07}). Using the union bound,
we found that the error probability is bounded
\[
{\small \begin{aligned} &P_{e}\leq \sum_{\mathcal{L},\mathcal{S\subseteq M}}2^{NR_{\mathrm{cf}}|\mathcal{L}|+N(\hat{ R}-R_{w})|\mathcal{S}|}\\ &\quad\quad\quad\cdot 2^{N(h(X_{\mathcal{L}},U_{\mathcal{S}}|X_{\mathcal{L}^{c}},U_{\mathcal{S}^{c}},T_{\mathcal{M}})-|\mathcal{L}|h(X)-|\mathcal{S}|h(U))}. \end{aligned}}%
\]
It follows that, in order to drive the probability of error to
zero, it is sufficient that
\begin{equation}
{\small \begin{aligned} &|\mathcal{L}|R_{\mathrm{cf}}+|\mathcal{S}|(\hat{R}-R_{w})\leq\\ &-h(X_{\mathcal{L}},U_{\mathcal{S}}|X_{\mathcal{L}^{c}},U_{\mathcal{S}^{c}},T_{\mathcal{M}})+ |\mathcal{L}|h(X_{m})+|\mathcal{S}|h(U_{m}). \end{aligned}}%
\label{ineq}
\end{equation}
Now, defining ${{\small \tilde{Y}_{m}=Y_{m}-T_{m}}}$ and ${{\small
\tilde {U}_{m}=\tilde{Y}_{m}^{(1)}+V_{m}}}$, and using the Markov
properties of the compression scheme, \ we have that
\begin{equation}
{\small \begin{aligned} I(Y_{m};U_{m}) &=h(U_{m})-h(U_{m}|Y_{m})\\ &=h(U_{m})-h(U_{m}|Y_{m},X_{\mathcal{M}},T_{\mathcal{M}})\\ &=h(U_{m})-h(\tilde{U}_{m}|\tilde{Y}_{m},X_{\mathcal{M}})\ , \end{aligned}}%
\label{one}
\end{equation}
and
\begin{equation}
{\small \begin{aligned} &h(X_{\mathcal{L}},U_{\mathcal{S}}|X_{\mathcal{L}^{c}},U_{\mathcal{S}^{c}},T_{ \mathcal{M}})=\\ &\quad=h(X_{\mathcal{L}}|X_{\mathcal{L}^{c}},U_{\mathcal{S} ^{c}},T_{\mathcal{M}})+h(U_{\mathcal{S}}|X_{\mathcal{M}},U_{\mathcal{S} ^{c}},T_{\mathcal{M}})\\ &\quad=h(X_{\mathcal{L}}|X_{\mathcal{L}^{c}},U_{\mathcal{S}^{c}},T_{\mathcal{M} })+|\mathcal{S}|h(U_{m}|X_{\mathcal{M}},T_{\mathcal{M}})\\ &\quad=h(X_{\mathcal{L}}|X_{\mathcal{L}^{c}},\tilde{U}_{\mathcal{S}^{c}})+| \mathcal{S}|h(\tilde{U}_{m}|X_{\mathcal{M}})\ , \end{aligned}}%
\label{two}
\end{equation}
and it is also easy to prove that
\begin{equation}
{\small |\mathcal{L}|h(X_{m})=h(X_{\mathcal{L}})=h(X_{\mathcal{L}
}|X_{\mathcal{L}^{c}})\ .}\label{three}
\end{equation}

Using \eqref{one}-\eqref{three} in \eqref{ineq} and dropping the
subscript denoting the cell index for symmetry, we get
\begin{equation*}\small\label{}
    |\mathcal{L}|R_{\mathrm{cf}}\leq|\mathcal{S}|(R_{\mathrm{w}}-I(\tilde{U};\tilde
{Y}|X_{\mathcal{M}
}))+I(X_{\mathcal{L}};\tilde{U}_{\mathcal{S}^{c}
}|X_{\mathcal{L}^{c}})\ ,
\end{equation*}
which corresponds to the result in
\cite{Sanderovich-Somekh-Shamai-ISIT07} by substitution of
${{\small \tilde{U}}}$ and ${{\small \tilde{Y}}}$ with ${{\small
U}}$ and ${{\small Y}}$, and the proof is completed by following
\cite{Sanderovich-Somekh-Shamai-ISIT07}.

\vspace{\myspace}
\section*{Acknowledgment}
\vspace{\myspace}

The research was supported by a Marie Curie Outgoing International
Fellowship and the NEWCOM++ network of excellence both within the 
6th and 7th European Community Framework Programmes, by the U.S. National
Science Foundation under Grants ANI-03-38807 and CNS-06-25637, and
the REMON consortium for wireless communication.


\end{document}